\documentclass[
reprint,
longbibliography,
superscriptaddress,
floatfix,
nofootinbib,
]{revtex4-2}

\usepackage{graphicx}
\usepackage{xcolor}
\usepackage[colorlinks=true,citecolor=blue,linkcolor=blue,urlcolor=blue]{hyperref}
\usepackage{amsmath,amssymb,amsthm,mathtools,braket}

\newtheorem{proposition}{Proposition}

\usepackage{algorithm}
\usepackage{algpseudocode}
\newcommand{\autorefalg}[2]{\hyperref[#1]{Algorithm~\ref{#1}(#2)}}
\usepackage[table]{xcolor}

\begin{document}

\title{Deterministic quantum trajectory via imaginary time evolution}

\author{Shivan Mittal}
\affiliation{Theoretical Division, Los Alamos National Laboratory, Los Alamos, New Mexico 87544, USA}
\affiliation{Department of Physics, University of Texas at Austin, Austin, Texas 78712, USA }

\author{Bin Yan}
\affiliation{Theoretical Division, Los Alamos National Laboratory, Los Alamos, New Mexico 87544, USA}

\date{\today}

\begin{abstract}
Stochastic quantum trajectories, such as pure state evolutions under unitary dynamics and random measurements, offer a crucial ensemble description of many-body open system dynamics. Recent studies have highlighted that individual quantum trajectories also encode essential physical information. Prominent examples include measurement-induced phase transitions, where a pure quantum state corresponding to fixed measurement outcomes (trajectories) exhibits distinct entanglement phases, depending on the measurement rate. However, direct observation of this effect is hindered by an exponential postselection barrier, whereby the probability of realizing a specific trajectory is exponentially small. We propose a deterministic method to efficiently prepare quantum trajectories in polynomial time using imaginary time evolution and, thus, overcome this fundamental challenge. We demonstrate that our method applies to a certain class of quantum states, and argue that universal approaches do not exist for any quantum trajectories. Our result paves the way for experimentally exploring the physics of individual quantum trajectories at scale and enables direct observation of certain postselection-dependent phenomena.
\end{abstract}

\maketitle

Quantum trajectories describe the stochastic evolution of pure quantum states. For example, a single spin precessing in a magnetic field while undergoing periodic projective measurements experiences random ``jumps'' that interrupt its otherwise unitary evolution~\cite{Weber2014Mapping}. Averaging over an ensemble of such trajectories provides an effective description of open quantum dynamics~\cite{Daley2014Quantum} arising from coupling to an external environment.

A \emph{single} quantum trajectory of a many-body system can also reveal critical information about the system. A striking example is the measurement-induced phase transition~\cite{Chan2019Unitary,Li2018Quantum,Skinner2019Measurement}, which has gathered considerable attention in recent years~\cite{Fisher2023Random,Noel2022Measurement,Google2023Measurement,Choi2020Quantum,Gullans2020Dynamical}. Consider a random quantum circuit where single-qubit midcircuit measurements are applied at a rate $p$~(Fig.~\ref{fig:illustration}). At long times (large circuit depth), the resulting pure state of the qubits exhibits two distinct phases, separated by a critical value of $p$. These phases are characterized by their entanglement patterns, either short range (following an area law) or long range (following a volume law). Crucially, such phase transitions manifest only at the level of individual quantum trajectories.

To physically observe these transitions, one must prepare multiple copies of the same trajectory, reproducing exactly the same sequence of measurement outcomes to gather statistical data. This presents a significant experimental challenge because there are $2^M$ possible measurement outcomes
for a circuit with $M$ single-qubit measurements. Thus, to observe a specific trajectory, the experiment must be repeated $\sim 2^M$ times. This exponential scaling constitutes the so-called postselection barrier, which severely limits the ability to probe single-trajectory physics at large scales~\cite{Gullans2020Scalable,Fisher2023Random}. The most recent direct experimental observation of measurement-induced phase transitions on IBM’s quantum computers~\cite{Koh2023Measurement} took seven months to run with just $14$ qubits---one of the most computationally intensive tasks ever performed on a quantum computer, highlighting the severity of this issue. Existing methods address this problem with indirect solutions~\cite{Li2023Cross,Dehghani2023Neural,Garratt2024Probing,Mcginley2024Postselection,Li2023Decodable,Passarelli2024Many,Buchhold2022Revealing,Buchhold2022Revealing}---for instance, by proposing alternative metrics, assuming simulability of quantum dynamics, or resorting to specific structures of the underlying physical systems. Efficiently and directly preparing multiple copies of a given quantum trajectory remains a long-standing challenge. It is even unclear whether such a task is fundamentally prohibited~\cite{Fisher2023Random}. 

\begin{figure}[b!]
   \includegraphics[width=\linewidth]{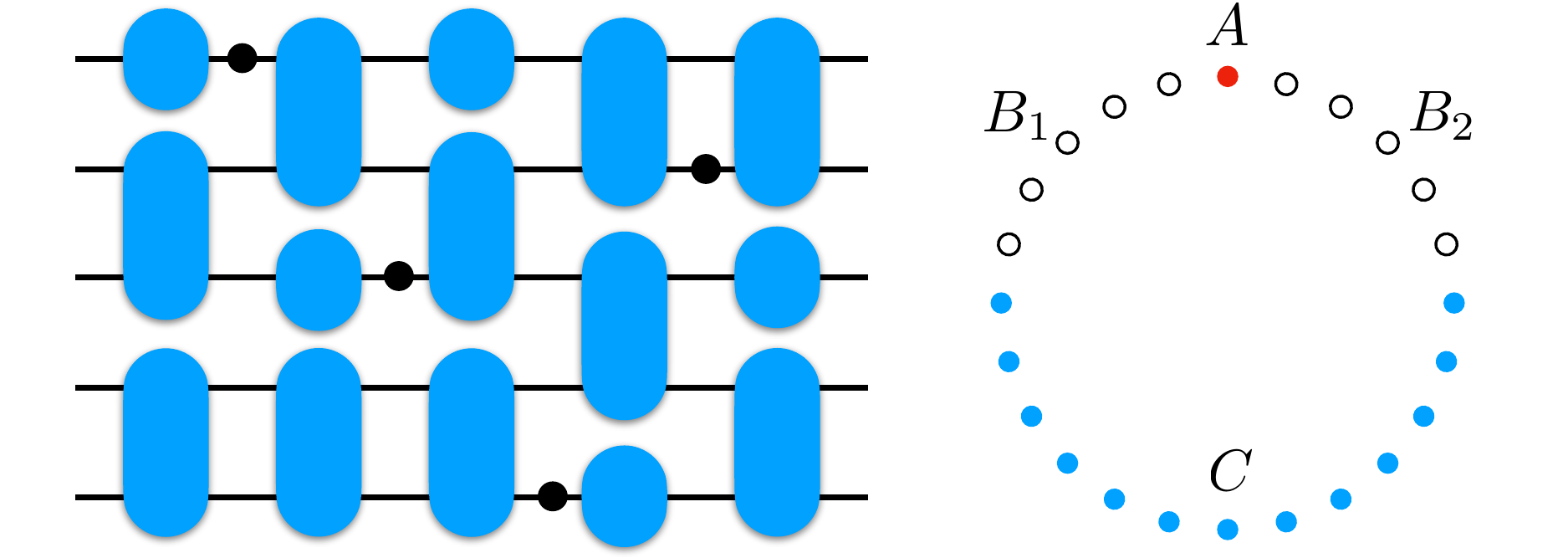}
    \caption{Left: a random quantum circuit is subject to single-qubit measurements (black dots) at a rate $p$. Depending on the value of $p$, the final state of the qubits can exhibit different phases with distinct entanglement patterns. For $p$ greater than the critical value $p_c$, the state of the qubits exhibits area law entanglement, i.e., the subsystem von Neumann entropy is proportional to the area of the boundary of the subsystem. Whereas for $p<p_c$, the entanglement satisfies the volume law. Right: we consider a one-dimensional qubit system in a periodic boundary condition. The cluster correlation can be quantified by the mutual information between a local subsystem $A$ and $C$, denoted by $I(A,C)(r)$, where $C$ is the cluster that includes all qubits whose distance to $A$ is greater than $r$.}  \label{fig:illustration}
\end{figure}

Here, we present a polynomial-time approach to deterministically prepare quantum trajectories of measured random quantum circuits relevant for measurement-induced phase transition in the area law regime. Our key innovation is leveraging imaginary time evolution to project a quantum state onto a target measurement outcome with certainty. We also argue that the existence of an efficient method for preparing any quantum trajectory would imply implausible complexity theoretic consequences; it would serve as
an extremely powerful resource for quantum computers by allowing them to solve any computational problems in probabilistic
polynomial time ($\mathsf{PP}$).

\emph{Imaginary time evolution---}To formulate the postselection problem in a general setting, consider an arbitrary quantum state $|\psi\rangle$ of $n$ qubits subjected to a local projective measurement, i.e., a measurement applied to a small subsystem of the qubits. Under a conventional measurement scheme, the many-body state collapses randomly into one of the possible postmeasurement states with probability given by Born’s rule. Our goal is to deterministically evolve $|\psi\rangle$ into a specific postmeasurement state, bypassing the inherent randomness of quantum measurement. 

Our approach is the following. We design a local Hamiltonian $H$ for the subsystem to be measured such that the ground state of $H$ corresponds to the desired postmeasurement state. For instance, to deterministically project a single target qubit into the state $|0\rangle$, we choose $H=\sigma_z$, where the Pauli operator $\sigma_z$ is supported on that qubit (in the convention that $|0\rangle$ is the ground state of $\sigma_z$). Then, we imaginary time evolve the initial state $|\psi\rangle$ for large imaginary time $\beta$. The resulting state
\begin{equation}\label{eq:imaginary}
     |\psi_\beta\rangle \equiv \mathcal{N}_\beta e^{-\beta H}|\psi\rangle
\end{equation}
converges to the exact target postmeasurement state in the limit $\beta$ tends to infinity. Here, $\mathcal{N}_\beta$ is the normalization factor consistent with the fact that imaginary time evolution does not preserve measurement outcome probabilities.

In practice, the state can only be evolved for finite $\beta$, thus, incurring an error to the ideal target state (which corresponds to infinite $\beta$). Fortunately, this error suppresses exponentially as $\beta$ increases linearly. The fidelity between the state at sufficiently large $\beta$ and the target state can be bounded from below by,
\begin{equation}\label{eq:fidelity}
\begin{aligned}
F(|\psi_\beta\rangle,|\psi_\infty\rangle) \ge 1 -  c e^{-2\beta \Delta},
\end{aligned}
\end{equation}
where $\Delta$ is the energy gap above the ground state of $H$ and $c \equiv (1-P)/P$, where $P$ is the probability of observing the target outcome if a physical measurement were performed. The proof of this lower bound is provided in Supplemental Material (SM). This exponential suppression of error is the key to the polynomial complexity of our deterministic postmeasurement state preparation method. 

Imaginary time evolution is a well-defined mathematical concept. However, contrary to unitary evolution, it cannot be directly realized in physical systems through Hamiltonian evolution. More precisely, since it is not even a quantum channel, there are no \emph{deterministic} physical processes that can realize~(\ref{eq:imaginary}) for any input state. Consequently, existing quantum algorithms for effectively implementing imaginary time evolution rely on postselecting certain outcomes. This imposes limitations on their applicability for our purposes because we cannot use those algorithms to deterministically prepare the target postmeasurement states that we seek if the imaginary time evolution itself is realized indeterministically. 

However, one might take a step back and look for a unitary evolution that effectively realizes~\eqref{eq:imaginary} that works only for a specific input state $|\psi\rangle$. That is, the unitary is $|\psi\rangle$-dependent. Such a deterministic quantum imaginary time evolution (DQITE) algorithm was developed in~\cite{Motta2020Determining}. The algorithm works as follows: suppose $H$ is a local Hamiltonian that applies nontrivially to a subsystem $A$ (see Fig.~\ref{fig:illustration}, right, for an illustration). One picks a larger regime $D$ that includes $A$ as its subsystem (the union of $A$ and $B_{1,2}$ in Fig.~\ref{fig:illustration}). A unitary $U_D$ with support on $D$ is guaranteed to exist that mimics the imaginary time evolution up to an error that is bounded by the correlation between $A$ and the complement of $D$, referred to as $C$ as shown in Fig.~\ref{fig:illustration}. Such correlations can be quantified by the mutual information between $A$ and $C$, denoted by $I(A,C)$. This is called the cluster correlation and is in contrast with the usual two-point correlation because $C$ is a cluster of local subsystems rather than a single local subsystem. The unitary $U_D$ is constructed via tomography on $D$. Crucially, as demonstrated in~\cite{Motta2020Determining}, for a polynomial complexity of the DQITE algorithm, the cluster correlation must decay exponentially fast in the distance $r$ between $A$ and $C$. This condition is fulfilled for the states produced in measured random quantum circuits with a measurement rate above the phase transition point, as we will discuss in the following.

\emph{Scaling analysis---}We have introduced a strategy for deterministically preparing a postmeasurement state via the DQITE algorithm when the quantum state satisfies the short-range cluster correlation property. Now, we demonstrate that it takes a polynomial in $M$ runtime to prepare a quantum trajectory specified by a total of $M$ measurement outcomes, each effectively realized via the procedure outlined above, up to $\epsilon$ error in trace distance.

For each deterministic measurement outcome preparation process, there are two sources of errors, one that stems from finite instead of infinite $\beta$, and one that comes from the effective simulation of the nonunitary imaginary time evolution via unitary algorithm. For simplicity, we demand that both the errors, quantified in trace distance, be the same and denote them by $\epsilon_\beta$. The final error that accumulates after $M$ postmeasurement state preparations via our procedure is upper bounded by $2\epsilon_\beta M$, which we demand be below $\epsilon$, thus
\begin{equation}\label{eq:betaerror}
    \epsilon_\beta \leq \frac{\epsilon}{2M}.
\end{equation}
To achieve this accuracy, we demand that fidelity in \eqref{eq:fidelity} be such that
\begin{equation}
    c^{1/2}e^{-\beta\Delta} \leq \epsilon_\beta,
\end{equation}
which, in turn, imposes a lower bound on the imaginary time $\beta$,
\begin{equation}\label{eq:beta}
    \beta \geq \frac{1}{\Delta}\log{\frac{c^{1/2}M}{\epsilon}}.
\end{equation}
The DQITE ``Trotterizes'' the imaginary time $\beta$ into multiple steps. Therefore, the runtime $T_\beta$ of the DQITE algorithm corresponding to preparing each measurement outcome also depends on the number of steps $n_\beta$ for Trotterizing $\beta$. For a fixed Trotter error, $n_\beta$ is proportional to $\beta$ and, thus, bounded through~\eqref{eq:beta}. According to~\cite{Motta2020Determining}, 
\begin{equation}\label{eq:runtime}
    T_\beta = \mathcal{O}(n^2_\beta/\epsilon_\beta) = \mathcal{O}(\beta^2/\epsilon_\beta).
\end{equation}

Bringing everything together, implementing the quantum imaginary time evolution algorithm for $M$ times with a final error $\epsilon$ requires a total runtime on the order of
\begin{equation}\label{eq:scaling}
\mathcal{O}\left(M^2\epsilon^{-1}\Delta^{-2}\log^2{(c^{1/2}M\epsilon^{-1})}\right).
\end{equation}
This concludes the polynomial scaling in $M$. We stress that this polynomial scaling relies on a polynomial sampling cost for performing tomography on subsystem $D$, requiring an exponentially decaying cluster correlation within the system and a finite Born probability $P$ of target postselection state. These conditions are fulfilled for the problem of measurement-induced phase transitions considered in this work. Rigorous analysis regarding these conditions is presented in Appendix B of SM.

We recall that $\Delta$ is the energy gap above the ground state of the Hamiltonian $H$ of the imaginary time evolution. For single-qubit projective measurements in the $Z$ basis, $H=\pm\sigma_Z$ and $\Delta = 2$. We also recall that $c = (1-P)/P$, where $P$ is the probability of observing the target outcome if a physical measurement were performed. An immediate remark is in order regarding the value of $c$; scaling~\eqref{eq:scaling} seems to suggest that for $n$ qubits the total runtime scales polynomially in $n$ even when $P$ is exponentially small in $n$ (when $c$ is exponentially large in $n$). This significantly contradicts our intuition that one cannot amplify an exponentially small signal in polynomial time. Resolution to this puzzle lies in the fact that the DQITE algorithm in~\cite{Motta2020Determining} requires local tomography to construct $U_D$---for an exponentially small $P$, one must perform an exponentially large number of measurements to achieve a finite precision. This fact was not discussed in~\cite{Motta2020Determining} and is not reflected in the runtime scaling~\eqref{eq:runtime}. For the case of local measurements and strongly correlated state $|\psi\rangle$, the order of $c$ is considered $\mathcal{O}(1)$. We numerically verified this condition for the problem of measurement-induced phase transitions in the following section.

\begin{figure}[t!]
    \centering
    \includegraphics[width=\linewidth]{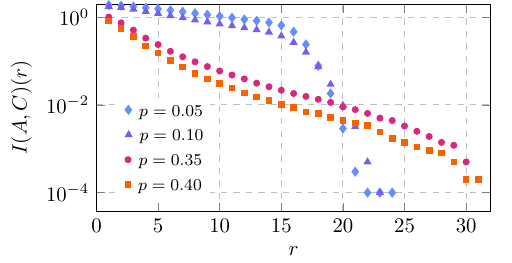}
    \caption{Decay of $I(A,C)(r)$ in the output state of $n=64$ qubit circuit with $L=n$ layers and at various measurement rates $p$. The critical measurement rate is $p_c = 0.16$~\cite{Li2019Measurement}.}
    \label{fig:cluster}
\end{figure}

\emph{Measurement-induced phase transitions---}In this section, we study the cluster correlations in the states produced by random circuits with midcircuit measurements, and show that they are short range for measurement rates greater than the critical value, thereby enabling the application of the deterministic postselection method we have developed.

We consider an even number $n$ of qubits. Let one brickwork layer of two-qubit gates be given by $U_{\text{odd}} U_{\text{even}}$, where $U_{\text{odd}} \equiv U_{1, 2} U_{3, 4} \cdots U_{n - 1, 0}$ and $U_{\text{even}} \equiv U_{0, 1} U_{2, 3} \cdots U_{n - 2, n - 1},$ and where $U_{i, j}$ denotes independent and identically distributed uniformly random two-qubit Cliﬀord gates~\footnote{We verified for small system sizes that all our conclusions carry over to measured random quantum circuits with Haar random 2-qubit gates, which is in accordance with the general expectation that exact unitary $2$-design gate sets give rise to the similar entanglement features output states of those circuit~\cite{Li2019Measurement}.} acting on qubits $i$ and $j$. The circuit consists of a total of $L$ such brickwork layers. Between any two layers each qubit is measured in the $Z$ basis with probability $p$ (referred to as the measurement rate). The structure is illustrated in Fig.~\ref{fig:illustration}.

We also consider periodic boundary condition. Suppose the system is divided into four regions: $A$, $B_1$, $B_2$ and $C$ (refer to Fig.~\ref{fig:illustration}, right), where $A$ is the single qubit subject to measurement, and $B_1$ and $B_2$ have the same size $r$, which can be viewed as the distance between $A$ and $C$. As discussed in the previous sections, to deterministically postselect qubit $A$, we implement imaginary time evolution $e^{\pm \beta \sigma_Z}$ using the DQITE algorithm, which requires local tomography in the union of $A$, $B_1$, and $B_2$. To guarantee a polynomial runtime, the mutual information between $A$ and $C$, $I(A,C)(r)$, must decay exponentially in $r$. This is not always true. For example, in Haar random states $I(A,C)(r)$ stays a constant until $r$ is greater than half of the total system, after which $I(A,C)(r)$ drops steeply to zero because correlations are scrambled over the entire state. The above intuition for Haar random states remains true for states generated by the random circuits with small measurement probabilities. For large $p$, when entanglement in the state satisfies an area law, the state becomes short-range-correlated and can result in an exponential decay of the mutual information.

In Fig.~\ref{fig:cluster}, we illustrate the decay in $I(A,C)(r)$ in the output state of the circuit at different measurement rates. As expected, $I(A,C)(r)$ decays exponentially for $p$ greater than a critical value. We can extract the said decay rate for the curves at various measurement rates $p$ and illustrate our findings. However, this approach gives an estimate of the decay rate that is prone to finite size effects of the small finite total system size. To mitigate finite size effects, in Fig.~\ref{fig:collapse}, we fix $r = n/16$ and compute $I(A,C)(r=n/16)$ for various total system size $n$. This allows us to perform a data collapse in order to extract the decay rate. The data collapse clearly reveals two distinct phases separated by a critical measurement rate $p_c$. For $p>p_c$, the cluster correlation decays exponentially. For $p<p_c$, a constant residual cluster correlation always exists. This residual correlation makes the DQITE algorithm inapplicable in this regime.

Figure~\ref{fig:fidelity} displays a numerical simulation of DQITE. Note that the DQITE breaks the Clifford structure. As a result we can only simulate the circuit with brute force, which significantly limits the size and depth of simulation. Therefore, here, we only present the deterministic postselection for a \emph{single} measurement and demonstrate that the procedure works with the error rate decaying as expected. Trajectories with repeated measurements can be deterministically implemented measurement by measurement. The full procedure is elaborated in Appendix C of SM. We take the $n=20$ qubit state $\ket{\psi}$ that results after $L = n$ layers of a measured random circuit at measurement rate $p=0.5 > p_c$ and $p=0.1 < p_c$. In fact, we employ Haar random two-qubit unitary gates for this set of numerics, since the DQITE algorithm involves gates that are outside the Clifford group. Then, we compute the infidelity between (i) the target postmeasurement state $\ket{\psi_{\infty}}$ after a single-qubit measurement on the first qubit of $\ket{\psi}$ and (ii) the state $\ket{\psi_{\beta}}$ produced by DQITE on $\ket{\psi}$ for increasing $r$. The simulations (Fig.~\ref{fig:fidelity}) demonstrate that in the area law region (states with exponentially decaying cluster correlations) increasing $r$ leads to faster convergence to the lower saturation values, while in the volume law region the fidelity remains small. This effectively allows us to separate two phases even though the deterministic protocol works only in the area law region.

\begin{figure}[t!]
    \centering
\includegraphics[width=\linewidth]{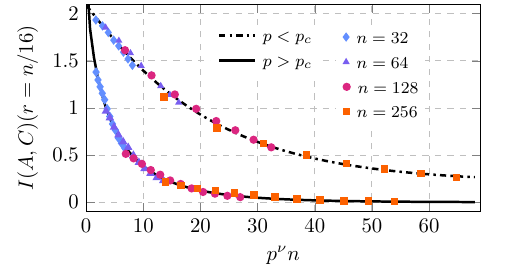}
    \caption{Cluster correlations at fixed $r=n/16$ for various values of measurement rates $p$ and number of qubits $n$. The data collapse parameter $\nu$ is obtained by minimizing the fitting error. For the area law entanglement regime, $\nu = 1.70$ with an exponential fitting (solid curve) $f(x)=2.68\exp{(-0.42x^{0.68})}$. For the volume law entanglement regime, $\nu = 0.75$ with an exponential fitting (dashed curve) $f(x)=1.88\exp{(-0.04x^{1.06}) + 0.21}$. The critical rate is $p_c=0.16$.}
    \label{fig:collapse}
\end{figure}

\emph{Universal postselection---}We have introduced an efficient deterministic postselection method for quantum states that satisfy exponential clustering. A natural question is whether such efficiency can also be achieved for generic quantum states. Throughout this discussion, we restrict attention to postmeasurement outcomes with \emph{finite} Born probabilities: if the desired outcome has an exponentially small probability, no deterministic postselection procedure can be efficient in general~\cite{PostBQP}.

Here, we show that efficient universal deterministic postselection procedures are highly unlikely to exist. In particular, we prove that the existence of such a universal algorithm implies the highly implausible collapse of complexity classes, 
$$\mathsf{BQP} = \mathsf{PP},$$
where $\mathsf{BQP}$ is bounded-error quantum polynomial time. The class $\mathsf{PP}$ is extraordinarily powerful: it contains $\mathsf{NP}$ and many other problems that are not believed to be efficiently solvable. Thus, the existence of an efficient universal postselection procedure would imply that quantum computers can efficiently solve all problems in $\mathsf{PP}$.

The key idea underlying our argument is a reduction from a generic $\mathsf{PP}$ problem to a sequence of postselections, each with finite Born probability. This reduction allows us to rule out the existence of an efficient universal procedure unless the unlikely complexity collapse above holds. The complete proof of the proposition is provided in Appendix D of SM.

\emph{Discussion---}Witnessing the measurement-induced phase transitions has faced significant experimental challenges due to the postselection barrier. The absence of direct solutions has led to the suspicion that this effect might be fundamentally unobservable, i.e., it must require an exponentially large overhead. Our results demonstrate the opposite. In particular, we demonstrate that quantum imaginary time evolution serves as an algorithm to efficiently and deterministically postselect measurement outcomes, provided that the states being measured satisfy the exponential clustering property. We demonstrate this fact in the physical scenario of measured random circuits that exhibit entanglement phases depending on measurement rate. This procedure allows us to physically construct multiple copies of states that are conditioned on midcircuit measurements without the ills of the postselection barrier. 

\begin{figure}[t!]
    \centering
\includegraphics[width=\linewidth]{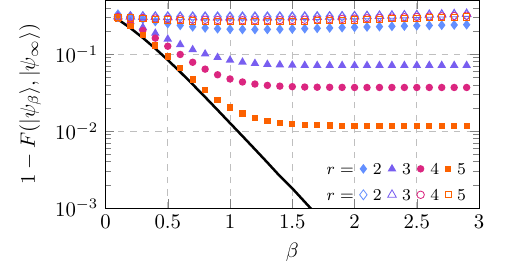}
    \caption{The infidelity $1 - F(\ket{\psi_{\beta}}, \ket{\psi_{\infty}})$ versus $\beta$ for various $r$. Here, $\ket{\psi_{\beta}}$ is the output of DQITE on $\ket{\psi}$ that is the output of a $n=20$ qubit measured random quantum circuits with $L=n$ layers with measurement rates $p=0.5$ (solid markers) and $p=0.1$ (empty markers). $\ket{\psi_{\infty}}$ is the post-measurement state upon measuring the first qubit in $\ket{\psi}$. For the $p=0.5$ case, the expected infidelity between $\ket{\psi_{\infty}}$ and the exactly (nonunitary) imaginary-time-evolved state to finite $\beta$ is shown as the black curve, confirming the exponential scaling in~\eqref{eq:fidelity}.}
    \label{fig:fidelity}
\end{figure}

The exponential clustering property alone (in 1D) suffices to demonstrate that the corresponding circuit is classically simulable~\cite{Brandao2015Exponential}. However, we emphasize that our goal is not to demonstrate the classical simulability of the measurement-induced phase transition in the area law phase. Instead, our objective is to demonstrate that we can \emph{physically} postselect on measurement outcomes. It is the latter that suffers from the exponential postselection barrier.

Small-scale proof-of-principle experiments of our algorithm on current noisy quantum processors are as feasible as recent demonstrations of DQITE~\cite{Motta2020Determining}. Importantly, the principal contribution of this work is the reduction of the postselection barrier, i.e., the exponential sampling cost associated with observing individual quantum trajectories, to a polynomial overhead in the regimes identified here. In practice, however, full-scale implementation faces several limitations: near-term devices constrain the achievable system size and the effective imaginary time. Thus, while near-term experiments can validate the key ingredients of the protocol and probe its scaling on small instances, large-scale realizations will likely require future hardware with significantly improved fidelities or early fault-tolerant capabilities~\cite{Preskill2025Beyond}.

Finally, under the mild complexity theoretic assumption that $\mathsf{BQP} \neq \mathsf{PP}$, we proved that an efficient deterministic postselection algorithm cannot exist for arbitrary quantum states. This opens several avenues for future study. For instance, it remains unclear whether certain nontrivial long-range-entangled states that violate exponential clustering might nevertheless admit efficient deterministic postselection. Another intriguing direction is whether the deterministic postselection developed here can be applied to other phenomena in which postselection poses a major challenge, such as deep thermalization~\cite{Ippoliti2022Solvable,Vairogs2025}.

\vspace{8pt}
\emph{Acknowledgment}---The authors thank Nick Hunter-Jones and Matteo Ippoliti for helpful discussions. This work was supported in part by the U.S. Department of Energy, Office of Science, Office of Advanced Scientific Computing Research, through the Quantum Internet to Accelerate Scientific Discovery Program, and in part by the LDRD program at Los Alamos. S.M.  also acknowledges support from the Center for Nonlinear Studies. B.Y. also acknowledges support from the U.S. Department of Energy, Office of Science, Basic Energy Sciences, under Early Career Award No.~DE-SC0024524.

\bibliography{references}

\vfill
\clearpage

\newpage

\appendix

\setcounter{page}{1}
\renewcommand\thefigure{\thesection\arabic{figure}}
\setcounter{figure}{0} 

\onecolumngrid

\begin{center}
\large{ Supplemental Material for \\ Deterministic quantum trajectory via imaginary time evolution}
\end{center}

\begin{center}
Shivan Mittal$^{1,2}$ and Bin Yan$^{1}$\\
\small{$^{1}$\textit{Theoretical Division, Los Alamos National Laboratory, Los Alamos, New Mexico 87545, USA}}\\
\small{$^{2}$\textit{Department of Physics, University of Texas at Austin, Austin, TX 78712, USA }}\\
\small{(Dated: \today)}
\end{center}

\section{End-to-end algorithm}

\algnewcommand\algiven{\textbf{Given:}}
\algnewcommand\Given{\item[\algiven]}
\algnewcommand\alprecomp{\textbf{Pre-compute:}}
\algnewcommand\Precomp{\item[\alprecomp]}
\algnewcommand\almain{\textbf{Main:}}
\algnewcommand\Main{\item[\almain]}
\renewcommand{\tablename}{Algorithm}

Here, we provide the pseudo code that outlines our end-to-end algorithm to deterministically prepare quantum trajectories. The algorithm is composed of two parts: 1) the main function that simulates the entire target trajectory with repeated unitary evolution and measurements (whose measurement outcomes are given) and 2) the deterministic post-selection subroutine used to prepare the post-measurement state of each individual measurement. The main function provides the big-picture chronology of events, and the subroutine is the quantum imaginary time evolution algorithm with respect to the Hamiltonian whose ground state is the post-measurement state.

In brief, in the pre-computation stage, we witness a quantum trajectory that we wish to deterministically, repeated prepare. Before we prepare our first deterministic sample, we iterate through the locations of measurements and learn unitaries that replicate the effect of post-selecting onto the witnessed trajectory of the pre-computation stage. After the end of the learning procedure, we apply all the learned unitaries [$O(\text{poly}(n))$ many] on the initial state to directly and deterministically prepare copies of the witnessed quantum trajectory.

%Our subroutine is simply a abstraction and paraphrasing of the deterministic quantum imaginary time evolution of Ref.~\cite{Motta2020Determining}. The number of samples $S$ over there is the accuracy to which the linear system is learned to then solve for the unitaries approximation of the imaginary time evolution.

\bigskip
\bigskip
\hrule
\begin{center}
{\textbf{Algorithm I: Deterministic Quantum Trajectories}}
\end{center}
\begin{algorithmic}[1]
    \Given number of qubits $n$, measurement rate $p$, number of layers $L$
    \Precomp Run 1 instance of measurement random circuit with the parameters $n, p$ and $L$ on any fixed initial state $\ket{\varphi}$, and store the description and locations of $2$-qubit gates and location and outcomes of mid-circuit measurements in classical memory. Initialize a list $G$ of gates with only unitary gates upto the first measurement in the circuit. Refer to list of measurement location and outcomes $\mathcal{M}:=\{(i, m_i)\}_{i}$.
    \Main
    \If{this is the first copy of the quantum trajectory}
        \For{$(i, m_i) \in \mathcal{M}$} 
        \State Perform all unitary operations in $G$ on the $\ket{\varphi}$.
        \State Perform Algorithm II($G, \ket{\varphi}, i, m_i$) to post-select on the $i^{\text{th}}$ measurement outcome.\footnote{The output of Algorithm II($G, \ket{\varphi}, i, m_i$) contains the unitary approximation quantum imaginary time evolution that implements the deterministic post-selection for the $i^{\text{th}}$ measurement.}
        \State Append the output of Algorithm II($G, \ket{\varphi}, i, m_i$) into $G$.
        \EndFor
    \EndIf\\
    \Return the state after all unitary operations in $G$ on $\ket{\varphi}$.
\end{algorithmic}

\bigskip
\bigskip
\hrule
\begin{center}
\textbf{Algorithm II: Deterministic Post-selection Subroutine}
\end{center}
    
    \begin{algorithmic}[1]
    \Given $G, \ket{\varphi}, i, m_i$  
    \Main
    \State Set $\text{count}=0$
    \While{$\text{count} < $number of samples $S$}
    \State Prepare a fresh sample of state $G\ket{\varphi}$
    \State Measure to perform tomography in the $\xi$-neighborhood of the $i^{\text{th}}$ qubit, where $\xi$ denotes the correlation length.
    \State $\text{count}=\text{count} + 1$
    \EndWhile
    \State Classically solve a linear systems $A\mathbf{x} = \mathbf{b}$ type problem for vectors of dimension $4^{2\xi}$ to determine the unitaries that approximately implement quantum imaginary time evolution on the $i^{\text{th}}$ qubit.\\
    \Return the list of unitaries that result from the previous step 
    \end{algorithmic}

\section{Bounding the failure probability}

As discussed in the main text, the total cost of the deterministic post-selection protocol is on the order of
\begin{equation}\label{eq:sm:scaling}
\mathcal{O}\left(M^2\epsilon^{-1}\Delta^{-2}\log^2{(c^{1/2}M\epsilon^{-1})}\right),
\end{equation}
where $M$ is the total number of single qubit projective measurements, $\epsilon$ is the desired final error tolerance, $\Delta$ is the energy gap above the ground state of the Hamiltonian $H$ of the imaginary time evolution (which is $2$ in our case). $c = (1-P)/P$, where $P$ is the Born probability of observing the target outcome if a physical measurement were performed.

The above estimate is based on an assumption that for each measurement, $P$---the probability of getting the desired the post-measurement state corresponding to a physical measurement being performed---is finite. However, in general this condition may not be satisfied. For instance, suppose a state before a measurement is $\ket{\psi}$. We would like to implement the protocol to deterministically project $|\psi\rangle$ to the post-measurement state $\ket{\psi_0}$ (corresponding to the measurement getting the outcome $0$). $P$ is the Born probability to witness the $0$ outcome. It could be that $P$ is extremely (exponentially) small. In this case, if we were to replace the measurement on $\ket{\psi}$ with the imaginary time evolution algorithm, it would take that algorithm exponentially long time to drive $\ket{\psi}$ to the post measurement state $\ket{\psi_0}$.

In our case, at each measurement position, the reduced density matrix of the qubit to be measured is very close to the maximally mixed state, therefore $P \sim 1/2$. However, in principle it is possible that for some instances
%(random circuit realizations) 
$P$ may take very small values.

One can quantify the distribution of $P$ and rigorously bound the probability that the protocol fails (in the sense that it needs an exponentially long runtime). However, the distribution of $P$ depends on both the particular circuit structure and the measurement outcomes of all previous measurements (it is a conditional probability). Moreover, even $P$ is highly concentrated to $1/2$ for the setup of measurement-induce phase transition as considered in this study, for some other scenarios where we would like to deterministically prepare quantum trajectories, $P$ may have a high probability to be very small, hindering the applicability of the proposed method.

Here, we provide a bound of the failure probability of our protocol that bypasses the above difficulties. The crucial observation is that the trajectory ($0$ and $1$ post-measurement state pattern) to be deterministically prepared is not manually determined, but is physically observed. That is, an experimentalist runs the experiment once and observes a single trajectory. We then use the proposed method to prepare this trajectory many times. Intuitively, \emph{if one trajectory contains some measurement outcomes whose probabilities are very small, then this trajectory is unlikely to be observed by the experimentalist in the first place}. This argument is formalized in the following.

Let $\delta$ be a threshold probability of our choice such that for a given single measurement outcome, if its probability $P$ is smaller than $\delta$, then the complexity of the imaginary time evolution algorithm is too high. Let
\begin{equation}
    \delta \equiv 1/\left[M\cdot{\rm poly}(n,M)\right],
\end{equation}
where ${\rm poly}(n,M)$ is an arbitrary polynomial of the number of qubits $n$ and the number of measurements $M$. 
If $P > \delta$, then the complexity of the protocol is still polynomial (verify by inserting $\delta$ in place of $P$ in the expression for $C$ in \autoref{prop:beta}, in Appendix~C. Otherwise we say the protocol fails (This is overkilling, therefore we are estimating a lose upper bound for the failure probability). 

For a given \emph{observed} trajectory of $M$ measurements, the probability that each single one of $M$ measurement outcomes has probability greater than $\delta$ is then
\begin{equation}
    (1-\delta)^M \approx 1- 1/{\rm poly}(n,M).
\end{equation}
Therefore, the failure probability---the probability that the experimentalist physically observes a trajectory, which contains at least one single measurement outcome of small probability (smaller than $\delta$ by our standard) to apply the deterministic post-selection method---is upper bounded by $1/{\rm poly}(n,M)$.

\section{Errors in imaginary time evolution}

Here, we prove the exponential suppression of error in imaginary time evolution as claimed in~\eqref{eq:fidelity} in the main text. For any quantum state $|\psi\rangle$, applying the imaginary time evolution $e^{-\beta H}$ for infinite $\beta$ projects it to the ground state subspace of $H$, given that $|\psi\rangle$ has a non-zero projection in that subspace.
Denote
\begin{equation}
     |\psi_\beta\rangle \equiv \mathcal{N}_\beta e^{-\beta H}|\psi\rangle,
\end{equation}
where $\mathcal{N}_\beta$ is the normalization factor. The following proposition quantifies the closeness of $|\psi_\beta\rangle$ when $\beta$ is finite to the state at infinite $\beta$.

\begin{proposition}\label{prop:beta}
For any $|\psi\rangle$, let $P$ denote the non-zero probability of $|\psi\rangle$ in the ground state subspace of $H$. The fidelity between $|\psi\rangle_\beta$ at sufficiently large $\beta$ and $|\psi\rangle_\infty$ is bounded by
\begin{equation}
\begin{aligned}
F(|\psi_\beta\rangle,|\psi_\infty\rangle) \ge 1 -  C e^{-2\beta \Delta},
\end{aligned}
\end{equation}
where $\Delta$ is the energy gap above the ground state of $H$ and $C = (1-P)/P$.
\end{proposition}
\begin{proof}
Without loss of generality, suppose the Hamiltonian $H$ applies non-trivially to subsystem $A$, that is, $H = H_A \otimes \mathbb{I}$. One can always decompose
\begin{equation}
    |\psi\rangle = \sum_{i=0}^{d_A-1} c_i |i\rangle_A|\phi_i\rangle,
\end{equation}
where $|i\rangle_A$ are eigenstates of $H_A$ with eigenvalues $E_i$, and $|0\rangle_A$ is the ground state; $d_A$ is the Hilbert space dimension of $A$. The imaginary time-evolved state is then
\begin{equation}
            |\psi_\beta\rangle = \mathcal{N}_\beta \sum_{i=0}^{d_A-1} c_i e^{-\beta E_i} |i\rangle_A|\phi_i\rangle,
\end{equation}
where the normalization factor is
\begin{equation}
    1/\mathcal{N}^2_\beta = \sum_{i=0}^{d_A-1} c^2_i e^{-2\beta E_i}.
\end{equation}
In particular,
\begin{equation}
        |\psi_\infty\rangle = |0\rangle_A|\phi_0\rangle.
\end{equation}

The fidelity between $|\psi_\beta\rangle$ and $|\psi_\infty\rangle$ can be evaluated as
\begin{equation}
F(|\psi_\beta\rangle,|\psi_\infty\rangle) = \mathcal{N}^2_\beta c_0^2 e^{-2\beta E_0},
\end{equation}
which can be bounded by
\begin{equation}
\begin{aligned}
1/F(|\psi_\beta\rangle,|\psi_\infty\rangle)
=& 1 + \sum_{i=1}^{d_A-1}\frac{c_i^2}{c_0^2}e^{-2\beta (E_i-E_0)}\\
\leq & 1 + e^{-2\beta \Delta} \sum_{i=1}^{d_A-1}\frac{c_i^2}{c_0^2}\\
\equiv & 1 + c e^{-2\beta \Delta}.
\end{aligned}
\end{equation}
Here we used $\Delta = E_1 - E_0$.

Therefore,
\begin{equation}
F(|\psi_\beta\rangle,|\psi_\infty\rangle) \geq \frac{1}{1 + C e^{-2\beta \Delta}} \approx 1 -  c e^{-2\beta \Delta},
\end{equation}
where the approximation is true when $\beta$ is sufficiently large.
The trace distance is bounded by
\begin{equation}
T(|\psi_\beta\rangle,|\psi_\infty\rangle) \leq  c^{1/2} e^{-\beta \Delta}.
\end{equation}

\end{proof}

\section{Non-existence of efficient universal solutions}

For an $n$-qubit quantum state $|\psi\rangle$, projective measurement in the $Z$ basis of a single qubit collapses $|\psi\rangle$ to the post-measurement state with the measured qubit in either $|0\rangle$ or $|1\rangle$. Suppose the probability of observing $|0\rangle$ (or $|1\rangle$) is finite, say, $\sim 1/2$. Are there processes that can deterministically prepare the post measurement state corresponding to a given measurement outcome, given $|\psi\rangle$ as an input? 

In the main text we demonstrated that one can use the deterministic quantum imaginary time evolution algorithm to achieve this with an accuracy $\epsilon$ in $\mathcal{O}(1/\epsilon)$ time, conditioned on that the state $|\psi\rangle$ satisfies the exponential clustering property. Here, we show that there does not exist such efficient processes that work for any state $|\psi\rangle$, unless $\mathsf{BQP} = \mathsf{PP}$. By efficiency we mean achieving an error $\epsilon$ in $\mathcal{O}(1/\epsilon^\gamma)$ time.

Suppose such a process exists and denote it as a quantum channel $\Lambda$. Let
\begin{equation}\label{eq:ppstate}
    |\phi\rangle = \alpha |0\rangle|\phi_0\rangle + \beta |1\rangle|\phi_1\rangle
\end{equation}
be a two-qubit state, where $\alpha/\beta \sim 2^{-n}$ is an exponentially small number. 

Note that we only assumed that $\Lambda$ can deterministically post-select a measurement outcome that has a finite probability of occurrence. Since the probability of post-measurement state $|0\rangle|\phi_0\rangle$ is exponentially small, it cannot be directly post-selected by $\Lambda$.

Consider a second register with $m$ qubits, each initialized in state $|0\rangle$. Applying the Control-Hadamard gates between the first qubit and each of the qubits in the second register prepares the total state into
\begin{equation}
     |\psi\rangle = \alpha |0\rangle|\phi_0\rangle |0\rangle^{\otimes m} + \beta |1\rangle|\phi_1\rangle |+\rangle^{\otimes m}.
\end{equation}
For this state, the probability of measuring $|0\rangle$ for any qubit in the second register is $\sim 1/2$. Thus, by assumption, we can apply $\Lambda$ to deterministically post-select this state. Suppose that $k$ qubits in the second register are post-selected to $|0\rangle$, the resulting total state becomes
\begin{equation}
     |\psi_k\rangle \propto \left(\alpha |0\rangle|\phi_0\rangle |0\rangle^{\otimes m-k} + \frac{1}{\sqrt{2^k}} \beta |1\rangle|\phi_1\rangle |+\rangle^{\otimes m-k}\right)|0\rangle^{k}.
\end{equation}
For this state, the probability of measuring $|0\rangle$ for any remaining qubit in the second register is 
\begin{equation}
    \frac{\alpha^2 + \beta^2/2^{k+1}}{\alpha^2 + \beta^2/2^k} \approx \frac{2^{-2n}+2^{-(k+1)}}{2^{-2n} + 2^{-k}},
\end{equation}
which is always greater than $1/2$. Thus, one can keep applying $\Lambda$ to all the remaining qubits one-by-one in the second register, until all the qubits in the second register are post-selected to $|0\rangle$. The resulting state is
\begin{equation}
     \left(\alpha |0\rangle|\phi_0\rangle  + \frac{1}{\sqrt{2^m}} \beta |1\rangle|\phi_1\rangle \right)|0\rangle^{m}.
\end{equation}
Assuming $m \gg n$ but still linear in $n$, one can thus achieve arbitrary precision in preparing the outcome $|0\rangle|\phi_0\rangle$ out of state~\eqref{eq:ppstate}, and simultaneously maintaining the overall polynomial runtime.

The above procedure allows us to use $\Lambda$ to deterministically post-selection the $0$ branch in state~\eqref{eq:ppstate}. However, applying Aaronson's algorithm in~[23] of the main text, this would further imply that one can solve any problem in $\mathsf{PP}$ by constructing a polynomial quantum circuit and efficiently post-selecting its output state. This implies the implausible consequence $\mathsf{BQP} = \mathsf{PostBQP} = \mathsf{PP}$.

%\subsection{New Figure 4}

%\begin{figure}[ht!]
%    \centering
%\includegraphics[width=0.5\linewidth]{fidelity_referee_a.pdf}
%    \caption{The infidelity $1 - F(\ket{\psi_{\beta}}, \ket{\psi_{\infty}})$ versus $\beta$ for various $r$. Here, $\ket{\psi_{\beta}}$ is the output of DQITE on $\ket{\psi}$ that is the output of a $n=20$ qubit measured random quantum circuits with $L=n$ layers and measurement rate $p=0.5$. $\ket{\psi_{\infty}}$ is the post measurement state upon measuring the first qubit in $\ket{\psi}$. Black curve corresponds to the infidelity between $\ket{\psi_{\infty}}$ and the exactly (not-unitary) imaginary-time-evolved state to finite $\beta$, confirming the exponential scaling in~\eqref{eq:fidelity}.}
%    \label{fig:fidelity_referee_a}
%\end{figure}

%\begin{figure}[ht!]
%    \centering
%\includegraphics[width=0.5\linewidth]{fidelity_referee_v.pdf}
%    \caption{The infidelity $1 - F(\ket{\psi_{\beta}}, \ket{\psi_{\infty}})$ versus $\beta$ for various $r$. Here, $\ket{\psi_{\beta}}$ is the output of DQITE on $\ket{\psi}$ that is the output of a $n=20$ qubit measured random quantum circuits with $L=n$ layers and measurement rate $p=0.1$. $\ket{\psi_{\infty}}$ is the post measurement state upon measuring the first qubit in $\ket{\psi}$. Black curve corresponds to the infidelity between $\ket{\psi_{\infty}}$ and the exactly (not-unitary) imaginary-time-evolved state to finite $\beta$, confirming the exponential scaling in~\eqref{eq:fidelity}.}
%    \label{fig:fidelity_referee_v}
%\end{figure}

\end{document}